\documentclass{new_tlp}

\usepackage{natbib}
\usepackage{times}  
\usepackage{helvet}  
\usepackage{courier}  
\usepackage{url}  
\usepackage{graphicx}  
\usepackage{amsmath}
\usepackage{footmisc}

\def\oo{\overline}

\def\ma{\mathcal{A}}
\def\mc{\mathcal{C}}
\def\mf{\mathcal{F}}
\def\mp{\ma}
\def\ml{\mathcal{L}}
\def\mq{\ml}
\def\mh{\mathcal{H}}
\def\ar{\leftarrow}
\def\rar{\rightarrow}
\def\lrar{\leftrightarrow}
\def\beq{\begin{equation}}
\def\eeq#1{\label{#1}\end{equation}}
\def\ba{\begin{array}}
\def\ea{\end{array}}
\def\ii#1{\hbox{\it #1\/}}

\def\no{\ii{not}}

\def\FLP{\hbox{\it FLP\/}}
\def\FT{\hbox{\it FT\/}}
\def\val{\hbox{\it val\/}}
\newtheorem{lemma}{Lemma}

\title{Relating Two Dialects of Answer Set Programming}
\author[Harrison and Lifschitz]
{AMELIA HARRISON\\
Google\\
\email{amelia.j.harrison@gmail.com}\\
\and
VLADIMIR LIFSCHITZ\\
University of Texas at Austin\\
\email{vl@cs.utexas.edu}
}
\begin{document}
\maketitle

\begin{abstract}
The input language of the answer set solver {\sc clingo} is based
  on the definition of a stable model proposed by Paolo Ferraris.
  The semantics of the ASP-Core language, developed by
  the ASP Standardization Working Group, uses the approach to stable models
  due to Wolfgang Faber, Nicola Leone, and Gerald Pfeifer.
  The two languages are based on
  different versions of the stable model semantics, and the ASP-Core
  document requires, ``for the sake of an uncontroversial semantics,'' that
  programs avoid the use of recursion through aggregates.
  In this paper we prove that the absence of recursion through aggregates
  does indeed guarantee the equivalence between the two versions of the
  stable model semantics, and show how that requirement can be relaxed
  without violating the equivalence property.  The paper is under
  consideration for publication in {\sl Theory and Practice of Logic
  Programming}.
\end{abstract}

\section{Introduction}
The stable model semantics of logic programs
serves as the semantic basis of answer set programming (ASP).
The ASP-Core document\footnote{\tt
 https://www.mat.unical.it/aspcomp2013/ASPStandardization.},
produced in 2012--2015 by the ASP Standardization Working Group,
was intended as a specification for the
behavior of answer set programming systems. The existence of
such a specification enables system comparisons and competitions to
evaluate such systems.

The semantics of ASP programs described in
that document differs from that of the input language of the widely used
answer set solver {\sc clingo}.\footnote{\tt http://potassco.org/clingo.}
The two languages are based on different versions of the stable model
semantics: the former on the FLP-semantics, proposed by
\cite{fab04,fab11} and generalized to arbitrary propositional formulas by
\cite{tru10}, and the latter on the approach of \cite{fer05}.

In view of this discrepancy, the ASP-Core document includes a warning:
``For the sake of an uncontroversial semantics, we require [the use of]
aggregates to be non-recursive'' (Section~6.3 of Version 2.03c).
Including this warning was apparently motivated by the belief that in the
absence of recursion through aggregates the functionality of {\sc
clingo} conforms with the ASP-Core semantics.

In this paper, that belief is turned into a theorem: for a
programming language that is essentially a large subset of
ASP-Core,\footnote{This language does not include classical negation, weak
constraints, optimize statements, and queries, and it does not allow
multiple aggregate elements within the same aggregate atom.  On the other
hand, it includes the symbols {\tt inf} and {\tt sup}
from the {\sc clingo}
language.} we prove
that the absence of recursion through aggregates guarantees the equivalence
between ASP-Core and {\sc clingo}.  Our theorem is actually stronger, in
two ways.  First, it shows that the view of recursion through
aggregates adopted in the ASP-Core document is unnecessarily broad when
applied to disjunctive programs (see Footnote~\ref{ftb}).  Second,
it shows that aggregates that do not contain negation as failure
can be used recursively without violating that property.  For example, the rule
\beq
\ba l
\verb|val(W,0) :- gate(G,and), output(W,G),|\\
\verb|            #count{WW : val(WW,0), input(WW,G)} > 0|
\ea
\eeq{gate}
which describes the propagation of binary signals through an and-gate
\cite[Example~9]{gel14a}, has the same meaning in both languages.

A few years ago it was difficult not only to prove such a theorem, but
even to state it properly, because a mathematically precise semantics of the
language of {\sc clingo} became available only with the publication by
\cite{geb15}.  The concept of a stable model for {\sc clingo} programs is
defined in that paper in two steps: first a transformation~$\tau$ is
introduced,\footnote{An oversight in the definiton of~$\tau$ in that
publication is corrected in the arXiv version of the paper, {\tt
http://arXiv.org/abs/1507.06576v2}.}  which turns a {\sc clingo} program into a
set of infinitary propositional formulas, and then the definition of a stable
model due to \cite{fer05}, extended to the infinitary case by \cite{tru12},
is invoked.  We will refer to stable models in the sense
of this two-step definition as ``FT-stable.''

To see why infinite conjunctions and disjunctions may be needed for
representing aggregate expressions, consider an instance of rule~(\ref{gate}):
\beq
\ba l
\verb|val(w1,0) :- gate(g1,and), output(w1,g1),|\\
\verb|             #count{WW : val(WW,0), input(WW,g1)} > 0|.
\ea
\eeq{gate1}
The expression
in the second line of~(\ref{gate1}) corresponds, informally speaking,
to an infinite disjunction:
for at least one of infinitely many possible values \verb|ww| of the
variable {\tt WW}, the stable model includes both
\verb|val(ww,0)| and \verb|input(ww,g1)|.

The semantics of ASP-Core programs is precisely defined in Section~2 of the
ASP-Core document, but that definition is not completely satisfactory: it is
not applicable to rules with local variables.  The problem is that the
definition of a ground instance in Section~2.2 of the document includes
replacing
the list $e_1;\dots;e_n$ of aggregate elements in an aggregate atom by its
instantiation $\hbox{inst}(\{e_1;\dots;e_n\})$; the instantiation, as
defined in the document, is an infinite object,
because the set of symbols that can be substituted for local variables
includes arbitrary integers and arbitrarily long symbolic constants.  For
example, the instantiation of the aggregare element
$$\verb|{WW : val(WW,0), input(WW,g1)}|$$
in the sense of the
ASP-Core document is an infinite object, because the set of symbols that can
be substituted for \verb|WW| is infinite.  So the
result of the replacement is not an ASP-Core program.
Prior to addressing the main topic of this paper, we propose a way to
correct this defect.  We use a two-step procedure, similar to the one
employed by \cite{geb15}: after applying a transformation~$\tau_1$, almost
identical to~$\tau$,\footnote{The original translation~$\tau$ could be used
for this purpose as well. However, the
definition of~$\tau_1$ seems more natural.\label{f1}} it refers to a
straightforward generalization of the definition of a stable
model due to \cite{fab04,fab11} to the infinitary case.
In the absence of local variables, this semantics is consistent with the
ASP-Core document \cite[Chapter~12]{har17b}.  Stable models in the sense
of this two-step definition will be called ``FLP-stable.''

We start by defining the syntax of programs, two versions of the stable
model semantics of infinitary formulas, and two
versions of the semantics of programs.
The main theorem
asserts that if the aggregates used in a program recursively do not contain
negation then the FLP-stable models of the program are the same
as its FT-stable models.
To prove the theorem we investigate under what conditions the
models of a set of infinitary propositional formulas that are stable in
the sense of Faber et al.~are identical to the models
stable in the sense of Ferraris and Truszczynski.

Results of this paper have been presented at the 17th International
Workshop on Non-Monotonic Reasoning.

\section{Syntax of Programs}

The syntax of programs is described here in an abstract fashion, in the
spirit of \cite{geb15}, so as to avoid inessential details related to
the use of ASCII characters.

We assume that three pairwise disjoint sets of symbols are selected:
{\sl numerals}, {\sl symbolic constants},  and
{\sl variables}.  Further, we assume that these sets do not contain the
symbols
\beq
+\qquad -\qquad \times\qquad /
\eeq{ops}
\vskip -.5cm
\beq
\ii{inf} \qquad \ii{sup}
\eeq{exts}
\beq
= \qquad \not = \qquad <\qquad >\qquad \leq\qquad \geq
\eeq{comps}
\beq
\no\qquad\land\qquad \lor\qquad\ar
\eeq{conns}
\beq
,\qquad
:\qquad (\qquad )\qquad \{\qquad \}
\eeq{punct}
and are different from the {\sl aggregate names} \ii{count}, \ii{sum},
\ii{max}, \ii{min}.  All these symbols together form the
alphabet of programs, and rules will be defined as strings over this alphabet.

We assume that a 1--1 correspondence between the set of numerals and the
set~${\bf Z}$ of integers is chosen.  For every integer~$n$, the corresponding
numeral will be denoted by~$\oo n$.

{\sl Terms} are defined recursively, as follows:
\begin{itemize}
\item all numerals, symbolic constants, and variables, as well as symbols
\eqref{exts} are terms;
\item if~$f$ is a symbolic constant and~${\bf t}$ is a non-empty tuple
of terms (separated by commas) then $f({\bf t})$ is a term;
\item if $t_1$ and $t_2$ are  terms and $\star$ is one of the
symbols~(\ref{ops}) then $(t_1\star t_2)$ is a  term.
\end{itemize}

A term, or a tuple of terms,  is {\sl ground}
if it does not contain variables. A term, or a tuple of terms, is {\sl precomputed} if it contains neither variables nor symbols~(\ref{ops}).
We assume a total order on precomputed terms such that \ii{inf} is its least
element, \ii{sup} is its greatest element, and,
for any integers $m$ and $n$, $\oo m \leq \oo n$ iff $m \leq n$.

For each aggregate name we define a function that maps every set of non-empty
tuples of precomputed terms to a precomputed term.
Functions corresponding to each of the aggregate names are defined below using
the following terminology.
If the first member of a tuple~$\bf t$
of precomputed terms is a numeral $\oo n$ then we say that the integer~$n$ is
the {\sl weight} of~$\bf t$; if $\bf t$ is empty or its first member is not
an numeral then the weight of~$\bf t$ is~0.
For any set~$T$ of tuples of precomputed terms,
\begin{itemize}
\item $\widehat{\ii{count}}(T)$ is the numeral corresponding to the cardinality
of~$T$ if~$T$ is finite, and $\ii{sup}$ otherwise;
\item $\widehat{\ii{sum}}(T)$ is the numeral corresponding to the sum of the
weights of all tuples in~$T$
if~$T$ contains finitely many tuples with non-zero weights, and $\ii{sup}$
otherwise;
\item $\widehat{\ii{min}}(T)$ is~$\ii{sup}$ if~$T$ is empty, the
least element of the set consisting of the first elements of the tuples in $T$
if $T$ is a finite non-empty set, and $\ii{inf}$ if~$T$ is infinite;
\item
$\widehat{\ii{max}}(T)$ is~$\ii{inf}$ if~$T$ is empty, the greatest element
of the set consisting of the first elements of the tuples in $T$
if $T$ is a finite non-empty set, and $\ii{sup}$ if~$T$ is infinite.
\end{itemize}

An {\sl atom} is a string of the form $p({\bf t})$
where $p$ is a symbolic constant and ${\bf t}$ is a tuple of terms.
For any atom~$A$, the strings
\beq
A\qquad\no\ A
\eeq{ls}
are {\sl symbolic literals}.
An {\sl arithmetic literal} is
a string of the form $t_1\prec t_2$
where~$t_1$,~$t_2$ are terms and $\prec$ is one of the symbols (\ref{comps}).
A {\sl literal} is a symbolic or arithmetic literal.\footnote{In the
parlance of the ASP-Core document, atoms are ``classical atoms,''
arithmetic literals are ``built-in atoms,'' and literals are ``naf-literals.''}

An {\sl aggregate atom} is a string of the form
\beq
       \alpha\{{\bf t}:{\bf L}\}\prec s,
\eeq{eq:ag}
where
\begin{itemize}
\item
$\alpha$ is an aggregate name,
\item
${\bf t}$ is a tuple of terms,
\item
${\bf L}$ is a tuple of literals called the ``conditions'' (if
${\bf L}$ is empty then the preceding colon may
be dropped),
\item
$\prec$ is one of the symbols~(\ref{comps}),
\item
and $s$ is a term.
\end{itemize}
For any aggregate atom~$A$, the strings~(\ref{ls})
are {\sl aggregate literals}; the former is called {\sl positive}, and the
latter is called {\sl negative}.


A {\sl rule} is a string of the form
\beq
H_1\,\lor\,\cdots\,\lor\,H_k \ar B_1\,\land\,\cdots\,\land\,B_n
\eeq{eq:rule}
($k,n\geq 0$), where each~$H_i$ is an atom, and each~$B_j$ is a literal or
aggregate literal.
The expression $B_1 \land \dots \land B_n$ is the
{\sl body} of the rule, and  $H_1 \lor \dots \lor H_k$ is the {\sl head}.
A {\sl program} is a finite set of rules.

About a variable we say that it is {\sl global}
\begin{itemize}
\item
in a symbolic or arithmetic literal $L$, if it occurs in~$L$;
\item
in an aggregate atom~\eqref{eq:ag} or its negation, if it occurs in~$s$;
\item
in a rule~\eqref{eq:rule}, if it is global in at least one of the
expressions~$H_i$,~$B_j$.
\end{itemize}
A variable that occurs in an expression but is not global in it is
{\sl local}. 

For example, in rule~(\ref{gate}), which is written as
$$\ba l
\ii{val}(W,\oo 0) \ar \ii{gate}(G,\ii{and})\land\ii{output}(W,G)\\
            \qquad\qquad\qquad\land\,
\ii{count}\{\ii{WW} : \ii{val}(\ii{WW},\oo 0), input(\ii{WW},G)\} > \oo 0
\ea$$
in the syntax described above, the variables $W$ and $G$ are global,
and \ii{WW} is local.

A literal or a rule is {\sl closed} if it has no global variables.

\section{Stable Models of Infinitary Formulas}

\subsection{Formulas} 

Let $\sigma$ be a propositional signature,
that is, a set of propositional atoms.  The sets
\hbox{$\mathcal{F}_0,\mathcal{F}_1,\ldots$} of formulas
are defined as follows:
\begin{itemize}
\item $\mathcal{F}_0=\sigma$,
\item $\mathcal{F}_{i+1}$ is obtained from $\mathcal{F}_{i}$ by
adding expressions $\mathcal{H}^\land$ and $\mathcal{H}^\lor$ for all subsets
$\mathcal{H}$ of $\mathcal{F}_i$, and expressions \hbox{$F\rar G$}
for all $F,G\in\mathcal{F}_i$.
\end{itemize}
The elements of $\bigcup^{\infty}_{i=0}\mathcal{F}_i$ are called {\sl
(infinitary propositional) formulas} over $\sigma$.

In an infinitary formula, $F\land G$ and $F\lor G$ are abbreviations for
$\{F,G\}^\land$ and $\{F,G\}^\lor$ respectively; $\top$ and $\bot$
are abbreviations for~$\emptyset^{\land}$ and $\emptyset^{\lor}$;
$\neg F$ stands for $F\rar\bot$, and $F\lrar G$
stands for \hbox{$(F\rar G)\land(G\rar F)$}.
{\sl Literals} over~$\sigma$ are atoms from~$\sigma$
and their negations.  If $\langle F_\iota\rangle_{\iota\in I}$ is a
family of formulas from one of the sets $\mf_i$ then the expression
$\bigwedge_\iota F_\iota$ stands for the formula
$\{F_\iota\,:\,\iota\in I\}^\land$,
and $\bigvee_\iota F_\iota$ stands for $\{F_\iota\,:\,\iota\in I\}^\lor$.

Subsets of a propositional signature~$\sigma$ will be called its
{\sl interpretations}.
The satisfaction relation between an interpretation and a formula is
defined recursively as follows:
\begin{itemize}
\item For every atom $p$ from $\sigma$, $I\models p$ if $p\in I$.
\item $I\models\mathcal{H}^\land$ if for every formula $F$ in~$\mathcal{H}$,
$I\models F$.
\item $I\models\mathcal{H}^\lor$ if there is a formula $F$ in~$\mathcal{H}$
such that $I\models F$.
\item $I\models F\rar G$ if $I\not\models F$ or $I\models G$.
\end{itemize}
We say that an interpretation satisfies a set $\mathcal{H}$ of formulas,
or is a {\sl model} of $\mathcal{H}$, if it satisfies every formula in
$\mathcal{H}$.  We say that~$\mathcal{H}$ {\sl entails} a formula~$F$ if
every model of~$\mathcal{H}$ satisfies~$F$.  Two sets of formulas are
{\sl  equivalent} if they have the same models.

\subsection{FLP-Stable Models} 

Let~$\mathcal{H}$ be a set of infinitary formulas of the form $G\rar H$,
where~$H$ is a disjunction of atoms from~$\sigma$.  The
{\sl FLP-reduct} $\FLP(\mathcal{H},I)$ of~$\mathcal{H}$ w.r.t.~an
interpretation~$I$ of~$\sigma$ is the set of all formulas~$G\rar H$
from~$\sigma$ such that~$I$ satisfies~$G$.
We say that~$I$ is an {\sl FLP-stable model} of $\mathcal{H}$
if it is minimal w.r.t.~set inclusion among the models of
$\FLP(\mathcal{H},I)$.

It is clear that~$I$ satisfies~$\FLP(\mathcal{H},I)$ iff~$I$
satisfies~$\mathcal{H}$.  Consequently every FLP-stable model of~$\mathcal{H}$
is a model of~$\mathcal{H}$.

\subsection{FT-Stable Models}

The {\sl FT-reduct} $\FT(F, I$) of an infinitary formula~$F$ w.r.t.~an
interpretation~$I$ is defined as follows:
\begin{itemize}
\item For any atom $p$ from $\sigma$,
$\FT(p, I)=\bot$ if $I\not\models p$; otherwise $\FT(p,I)=p$.
\item $\FT(\mathcal{H}^\land, I)=\{\FT(G,I) \ |\ G\in\mathcal{H}\}^\land$.
\item $\FT(\mathcal{H}^\lor, I)=\{\FT(G,I) \ |\ G\in\mathcal{H}\}^\lor$.
\item $\FT(G\rar H, I)=\bot$ if $I\not\models G\rar H$;
otherwise
\hbox{$\FT(G\rar H, I)=\FT(G, I)\rar \FT(H,I)$}.
\end{itemize}
The FT-reduct $\FT(\mathcal{H},I)$ of a set~$\mathcal{H}$ of formulas is
defined as
the set of the reducts $\FT(F,I)$ of all formulas~$F$ from~$\mathcal{H}$.
An interpretation~$I$ is an {\sl FT-stable model} of $\mathcal{H}$
if it is minimal w.r.t.~set inclusion among the models of
$\FT(\mathcal{H},I)$.

It is easy to show by induction that~$I$ satisfies~$\FT(F,I)$ iff~$I$
satisfies~$F$.  Consequently every FT-stable model of a set of formulas
is a model of that set.

It is easy to check also that if~$I$ does not satisfy~$F$ then
$\FT(F,I)$ is  equivalent to~$\bot$.

\subsection{Comparison} 

An FLP-stable model of a set of formulas is not necessarily FT-stable, and
an FT-stable model is not necessarily FLP-stable.  For example, consider
(the singleton set containing) the formula
\beq
p \lor \neg p \rar p.
\eeq{eq:nodtight}
It has no FT-stable models, but the interpretation $\{p\}$ is its
FLP-stable model. On the other hand, the formula
\beq
\neg \neg p \rar p
\eeq{eq:nodtight2}
has two FT-stable models, $\emptyset$ and $\{p\}$, but latter is not
FLP-stable.

It is clear that replacing the antecedent of an implication by an
 equivalent formula within any set of formulas does not affect
its FLP-stable models.  For instance, from the perspective of the FLP
semantics, formula~(\ref{eq:nodtight}) has the same meaning as
$\top\rar p$, and~(\ref{eq:nodtight2}) has the same meaning as
$p\rar p$.  On the other hand,
the FLP-stable models may change if we break an implication of the form
$F\lor G \rar H$ into $F\rar H$ and $G\rar H$.  For instance,
breaking~(\ref{eq:nodtight}) into~$p\rar p$ and $\neg p\rar p$ gives a
set without FLP-stable models.

With the FT semantics, it is the other way around: it does matter, generally,
whether we write $\neg\neg p$ or~$p$ in the antecedent of an
implication, but breaking $F\lor G \rar H$ into two implications cannot
affect the set of stable models.

Transformations of infinitary formulas that do not affect their FT-stable
models were studied by \cite{har17}.  These authors extended, in
particular, the logic of here-and-there introduced by
\cite{hey30} to infinitary propositional formulas and showed that
any two sets of infinitary formulas that have the same models in the
infinitary logic of here-and-there have also the same FT-stable models.

\section{Semantics of Programs} 

In this section, we define two very similar translations,~$\tau_1$
and~$\tau$.  Each of them transforms any program into a set of infinitary
formulas over the signature~$\sigma_0$ consisting of all atoms of the
form $p({\bf t})$, where $p$ is a symbolic constant and~${\bf t}$
is a tuple of precomputed terms.  The definition of~$\tau$ follows
\cite{geb15}.

Given these translations, the two versions of the semantics of programs are
defined as follows.  The {\sl FLP-stable models} of a program~$\Pi$ are the
FLP-stable models of~$\tau_1\Pi$.  The {\sl FT-stable models} of~$\Pi$ are the
FT-stable models of~$\tau\Pi$.

\subsection{Semantics of Terms}

The semantics of terms tells us, for every ground term~$t$, whether it
is {\sl well-formed}, and if it is, which precomputed term is considered its
{\sl value}:\footnote{In the input language of
{\sc clingo}, a term may contain ``intervals'',
such as $1..3$, and in that more general setting a ground term may have
several values.}
\begin{itemize}
\item If $t$ is a numeral, symbolic constant, or one of the symbols
$\ii{inf}$ or $\ii{sup}$ then~$t$ is well-formed, and its value $\val(t)$ is
$t$ itself.
\item If $t$ is $f(t_1, \dots, t_n)$ and the terms $t_1, \dots, t_n$ are
well-formed, then $t$ is well-formed also, and $\val(t)$ is
$f(\val(t_1),\dots,\val(t_n))$.
\item If $t$ is $(t_1 + t_2)$ and the values of $t_1$ and $t_2$ are
numerals~$\oo{n_1}$,~$\oo{n_2}$ then~$t$ is well-formed, and $\val(t)$ is
$\oo{n_1+n_2}$; similarly when $t$ is $(t_1 - t_2)$ or $(t_1 \times t_2)$.
\item If $t$ is $(t_1 / t_2)$, the values of $t_1$ and $t_2$ are
numerals~$\oo{n_1}$,~$\oo{n_2}$, and $n_2\neq 0$ then~$t$ is well-formed, and
$\val(t)$ is $\oo{\lfloor n_1/n_2 \rfloor}$.
\end{itemize}
For example, the value of $\oo 7/\oo 3$ is $\oo 2$; the terms $\oo 7/\oo 0$
and $\oo 7/a$, where $a$ is a symbolic constant, are not well-formed.

If ${\bf t}$ is a tuple $t_1, \dots, t_n$ of well-formed ground terms
then we say that $\bf t$ is well-formed, and its value $\val({\bf t})$ is
the tuple $\val(t_1),\dots,\val(t_n)$.

A closed arithmetic literal $t_1\prec t_2$ is well-formed if~$t_1$ and~$t_2$
are well-formed.
A closed symbolic literal $p({\bf t})$ or $\no\ p({\bf t})$ is well-formed
if $\bf t$ is well-formed.  A closed aggregate literal~$E$ or $\no\ E$, where~$E$
is~(\ref{eq:ag}), is well-formed if $s$ is well-formed.

\subsection{Semantics of Arithmetic and Symbolic Literals} 

A well-formed arithmetic literal $t_1\prec t_2$
is {\sl true} if $\val(t_1)\prec\val(t_2)$, and {\sl false} otherwise.

The result of applying the transformation~$\tau_1$ to a well-formed
symbolic literal is defined as follows:
$$\tau_1(p({\bf t}))\hbox{ is }p(\val({\bf t}));
\quad\tau_1(\no\ p({\bf t}))\hbox{ is }\neg p(\val(\bf t)).$$

About a tuple of well-formed literals we say that it is {\sl nontrivial} if
all arithmetic literals in it are true.
If~$\bf L$ is a nontrivial tuple of well-formed arithmetic and symbolic
literals then $\tau_1{\bf L}$ stands for the
conjunction of the formulas~$\tau_1 L$ for all symbolic literals~$L$
in~${\bf L}$.

\subsection{Semantics of Aggregate Literals}

Let~$E$ be a well-formed aggregate atom~(\ref{eq:ag}), and let~${\bf x}$
be the list of variables occurring in~${\bf t}:{\bf L}$.  By~$A$
we denote the set of all tuples~${\bf r}$ of precomputed terms of the same
length as~${\bf x}$ such that
\begin{itemize}
\item[(i)]
${\bf t}^{\bf x}_{\bf r}$ is well-formed, and
\item[(ii)]
${\bf L}^{\bf x}_{\bf r}$ is well-formed and nontrivial.\footnote{Here
${\bf t}^{\bf x}_{\bf r}$ stands for the result of substituting~${\bf r}$
for~${\bf x}$ in $\bf t$.  The meaning of ${\bf L}^{\bf x}_{\bf r}$ is
similar.}
\end{itemize}
For any subset~$\Delta$
of $A$, by $\val(\Delta)$ we denote
the set of tuples ${\bf t}^{{\bf x}}_{\bf r}$ for all
${\bf r}\in\Delta$.
We say that $\Delta$ {\sl justifies}~$E$ if the relation $\prec$ holds
between $\widehat \alpha (\val(\Delta))$ and $\val(s)$.
We define $\tau_1 E$ to be the disjunction of formulas
\beq
\bigwedge_{{\bf r} \in \Delta} \tau_1({\bf L}^{\bf x}_{\bf r}) \land
\bigwedge_{{\bf r} \in  A\setminus\Delta}\neg\tau_1({\bf L}^{\bf x}_{\bf r})
\eeq{eq:altagtrans}
over the subsets $\Delta$ of $A$ that justify $E$.

Assume, for example, that~$E$ is
\beq\ii{count}\{X:p(X)\}=\oo 0.
\eeq{example}
Then

\begin{itemize}
\item $\bf t$ is $X$, $\bf L$ is $p(X)$, $\bf x$ is $X$,
and $A$ is the set of all precomputed terms,
$\val(\Delta)$ is $\Delta$;
\item $\widehat\alpha(\val(\Delta))$ is the cardinality of~$\Delta$ if $\Delta$ is
finite and $\ii{sup}$ otherwise;
\item $\Delta$ justifies~(\ref{example}) iff
$\Delta=\emptyset$.
\end{itemize}
It follows that $\tau_1 E$ is in this case the conjunction of the
formulas $\neg p(r)$ over all precomputed terms~$r$.
\medskip

The result of applying $\tau_1$ to a negative aggregate literal $\no\ E$
is $\neg \tau_1 E$.

The definition of $\tau_1{\bf L}$ given earlier
can be extended now to nontrivial tuples that may include well-formed
literals of all three kinds: for any
such tuple~$\bf L$, $\tau_1{\bf L}$ stands for the
conjunction of the formulas~$\tau_1 L$ for all symbolic
literals and aggregate literals~$L$ in~${\bf L}$.

\subsection{Applying $\tau_1$ to Rules and Programs} 

The result of applying $\tau_1$ to a rule~(\ref{eq:rule}) is defined as
the set of all formulas of the form
\beq
\tau_1((B_1,\dots,B_n)^{\bf x}_{\bf r}) \rar
\tau_1{(H_1)}^{\bf x}_{\bf r}\,\lor\,\cdots\,\lor\,
\tau_1{(H_k)}^{\bf x}_{\bf r}
\eeq{trr1}
where $\bf x$ is the list of all global variables of the rule, and $\bf r$
is any tuple of precomputed terms of the same length as~{\bf x}
such that $(B_1,\dots,B_n)^{\bf x}_{\bf r}$ is nontrivial and all literals
${(H_i)}^{\bf x}_{\bf r}$ are well-formed.

For example, the result of applying~$\tau_1$ to the rule
$$q(X/Y) \ar p(X,Y) \land X>Y$$
is the set of all formulas of the form
$$p(\oo m,\oo n) \rar q(\oo{\lfloor m/n \rfloor})$$
where $m$, $n$ are integers such that $m>n$ and $n\neq 0$.

For any program~$\Pi$, $\tau_1 \Pi$ stands for the union of the sets~$\tau_1 R$
for all rules~$R$ of~$\Pi$.

\subsection{Transformation~$\tau$}

The definition of~$\tau$ differs from the definition of~$\tau_1$
in only one place: in the treatment of aggregate atoms.
In the spirit of \cite{fer05},
we define $\tau E$ to be the conjunction of the implications
\beq
  \bigwedge_{{\bf r}\in\Delta}\tau({\bf L}^{{\bf x}}_{\bf r})
  \,\rar\,
  \bigvee_{{\bf r}\in A\setminus\Delta}\tau({\bf L}^{{\bf x}}_
  {\bf r})
\eeq{eq:agdef}
over the subsets $\Delta$ of $A$ that do not justify $E$.

For example, if~$E$ is~(\ref{example}) then $\tau E$ is
$$\bigwedge_{\Delta\subseteq A,\ \Delta\neq\emptyset}\left(
\bigwedge_{r\in\Delta}p(r) \rar \bigvee_{r\in A\setminus\Delta} p(r)
\right).$$

It is easy to show that $\tau E$ is  equivalent to $\tau_1 E$.
Consider the disjunction $D$ of formulas~\eqref{eq:altagtrans} over all
subsets $\Delta$ of $A$ that do not justify $E$. It is easy to see that every
interpretation satisfies either $\tau_1 E$ or $D$. On the other hand, no
interpretation satisfies both $D$ and $\tau_1 E$, because in every disjunctive
term of $\tau_1 E$ and every disjunctive term of $D$ there is a pair of
conflicting conjunctive terms. It follows that~$D$ is
 equivalent to $\neg \tau_1 E$. It is clear that $D$ is also
 equivalent to $\neg \tau E$.

Since all occurrences of translations~$\tau_1 E$ in implication~(\ref{trr1})
belong to its antecedent, it follows that $\tau$ could be used instead
of~$\tau_1$ in the definition of an FLP-stable model of a program.
For the definition of an FT-stable model of a program,
however, the difference between~$\tau_1$ and~$\tau$ is essential. Although
the translation $\tau_1$ will not be used in the statement or proof of the main
theorem, we introduce it here because it is simpler than $\tau$ in the sense
that in application to aggregate literals it does not produce implications. We
anticipate that for establishing other properties of FLP-stable models it
may be a useful tool.

\section{Main Theorem} 

To see that the FLP and FT semantics of programs are generally not
equivalent, consider the one-rule program
 \beq
p \ar \ii{count}\{\oo 1: not \; p\} < \oo 1.
 \eeq{eq:progd}
The result of applying $\tau$ to this program is
$\neg \neg p \rar p$.
The FT-stable models are $\emptyset$ and $\{p\}$; the first of them is an
FLP-stable model, and the second is not.

Our main theorem gives a condition ensuring that the FLP-stable models and
FT-stable models of a program are the same.
To state it, we need to describe the precise meaning of
``recursion through aggregates.''

The {\sl predicate symbol} of an atom $p(t_1, \dots, t_n)$ is the pair~$p/n$.
The {\sl predicate dependency graph} of a program $\Pi$ is the directed
graph that
\begin{itemize}
\item has the predicate symbols of atoms occurring in $\Pi$ as its vertices,
and
\item has and edge from $p/n$ to $q/m$ if there is a rule $R$ in $\Pi$
such that
$p/n$ is the predicate symbol of an
atom occurring in the head of $R$, and $q/m$ is the predicate symbol of an
atom occurring in the body of~$R$.\footnote{The definition of the
predicate dependency graph in the ASP-Core document includes also
edges between predicate symbols of atoms occurring in the head of the same
rule.  Dropping these edges from the graph makes the assertion of the main
theorem stronger.  By proving the stronger version of the theorem we
show that the understanding of recursion through aggregates in the ASP-Core
document is unnecessarily broad.\label{ftb}}
\end{itemize}

We say that an occurrence of an aggregate literal $L$ in a rule $R$ is
{\sl recursive} with respect to a program $\Pi$ containing~$R$ if
for some predicate symbol $p/n$ occurring in $L$ and some predicate
symbol~$q/m$ occurring in the head of~$R$ there exists a path from
$p/n$ to $q/m$ in the predicate dependency graph of $\Pi$.

For example, the predicate dependency graph
of program~\eqref{eq:progd} has a single vertex~$p/0$ and an edge from~$p/0$
to itself. The aggregate literal in the body of this program is recursive.
Consider, on the other hand, the one-rule program
$$
q \ar \no \; \ii{count}\{\oo 1:p\} < \oo 1.
$$
Its predicate dependency graph has the vertices $p/0$ and $q/0$, and an
edge from $q/0$ to $p/0$.  Since there is no path from $p/0$ to $q/0$ in
this graph, the aggregate literal in the body of this rule is not recursive.

We say that an aggregate literal is {\sl positive} if
it is an aggregate atom and all symbolic literals occurring in it
are positive.

\medskip\noindent
{\em Main Theorem}

\noindent
If every aggregate literal that is recursive with respect to a program~$\Pi$
is positive then the FLP-stable models of~$\Pi$ are the same as the
FT-stable models of~$\Pi$.

\medskip
In particular, if all aggregate literals in~$\Pi$ are positive then~$\Pi$ has
the same FLP- and FT-stable models.  For example, consider the one-rule program
$$
p \ar \ii{count}\{\oo 1:p\} > \oo 0.
$$
The only aggregate literal in this program is positive;
according to the main theorem, the program has the same FLP- and
 FT-stable models.  Indeed,
it is easy to verify that $\emptyset$ is the only FLP-stable model of
this program and also its only FT-stable model.

\section{Main Lemma} 

In this section we talk about infinitary formulas over an arbitrary
propositional signature~$\sigma$.

Formulas~$p$,~$\neg p$, $\neg\neg p$, where~$p$ is an atom from~$\sigma$,
will be called {\sl extended literals}.  A {\sl simple disjunction} is a
disjunction of extended literals.
A {\sl simple implication} is an implication \hbox{$\mp^\land \rar \mq^\lor$}
such that its antecedent~$\mp^\land$ is a conjunction of atoms and its
consequent~$\mq^\lor$ is a simple disjunction.
A conjunction of simple implications will be
called a {\sl simple formula}.  Formulas of the form $G\rar H$,
where $G$ is a simple formula and $H$ is a disjunction of atoms, will be
called {\sl simple rules}.\footnote{Note that a simple rule is not a rule
in the sense of the programming language described above;
it is an infinitary propositional formula of a
special syntactic form.}
A {\sl simple program} is a set of simple rules.

For example,~(\ref{eq:nodtight}),~(\ref{eq:nodtight2}) can be rewritten as
simple rules
\beq
(\top\rar p \lor \neg p) \rar p,
\eeq{eq:n}
\beq
(\top\rar \neg \neg p) \rar p.
\eeq{eq:n2}
In the proof of Main Theorem 
we will show how any
formula obtained by applying transformation~$\tau$ to a program
can be transformed into a simple rule with the same meaning.

In the statement of Main Lemma below, we refer to simple programs that are
``FT-tight'' and ``FLP-tight.'' The lemma asserts that if a program is
FT-tight then its FLP-stable models are FT-stable;  if a program is
FLP-tight then its FT-stable models are FLP-stable.  To describe these
two classes of simple programs we need the following preliminary
definitions.

An atom~$p$ {\sl occurs strictly positively} in a simple formula~$F$ if
there is a conjunctive term $\mp^\land \rar \mq^\lor$ in~$F$ such that~$p$
belongs to $\mq$.  An atom~$p$ {\sl occurs positively} in a simple
formula~$F$ if there is a conjunctive term $\mp^\land \rar \mq^\lor$ in~$F$
such that~$p$ or $\neg \neg p$ belongs to $\mq$.

We define the {\sl (extended
positive) dependency graph} of a simple
program~$\mathcal{H}$ to be the graph that has
\begin{itemize}
\item all atoms occurring in $\mathcal{H}$ as its vertices, and
\item an edge from $p$ to $q$ if for some formula $G \rar H$
in $\mathcal{H}$,~$p$ is a disjunctive term in $H$ and $q$ occurs positively in $G$.
\end{itemize}
For example, the simple programs~(\ref{eq:n}),~(\ref{eq:n2})
have the same dependency graph: a self-loop at~$p$.\footnote{We call the
graph {\sl extended} positive to emphasize the fact that the definition
reqires~$q$ to occur positively in~$G$, but not strictly positively.}

A simple implication $\mp^\land \rar \mq^\lor$ will be called {\sl positive} if
$\mq$ is a set of atoms, and {\sl non-positive} otherwise. An edge from
$p$ to $q$ in the dependency graph of a simple program $\mh$ will be called
{\sl FT-critical} if for some formula $G \rar H$ in $\mathcal{H}$,~$p$ is
a disjunctive term in $H$ and $q$ occurs strictly positively in
some non-positive conjunctive term $D$ of~$G$.
We call a simple program {\sl FT-tight} if its dependency graph has no path
containing infinitely many FT-critical edges.\footnote{In the case of a finite
dependency graph, this condition is equivalent to requiring that no cycle
contains an FT-critical edge.}

Consider, for example, the dependency graph of program~\eqref{eq:n}.  Its
only edge---the self-loop at~$p$---is FT-critical, because the implication
$\top\rar p\lor\neg p$ is non-positive, and~$p$ occurs strictly positively
in it.  It follows that the program is not FT-tight: consider the path
consisting of infinitely many repetitions of this self-loop.  On the other
hand, in the dependency graph of program~\eqref{eq:n2} the same edge
is not FT-critical, because~$p$ does not occur strictly positively
in the implication $\top\rar\neg\neg p$.  Program~\eqref{eq:n2} is FT-tight.

An edge from $p$ to $q$ in the dependency graph of a simple program $\mh$
will be called {\sl FLP-critical} if for some simple rule $G \rar H$
in $\mathcal{H}$,~$p$ is a disjunctive term in $H$ and, for some
conjunctive term $\mp^\land \rar \mq^\lor$ of~$G$, $\neg \neg q$ belongs
to~$\mq$.
We call a simple program {\sl FLP-tight} if its dependency graph has no path
containing infinitely many FLP-critical edges.

It is clear that if there are no extended literals of the form~$\neg\neg p$
in a simple program then there are no FLP-critical edges in its dependency
graph, so that the program is FLP-tight.  For example,~\eqref{eq:n} is
a simple program of this kind.  On the other hand, in the dependency graph
of program~\eqref{eq:n2} the self-loop at~$p$ is FLP-critical, so that the
program is not FLP-tight.

\medskip\noindent
{\em Main Lemma}

\noindent
For any simple program~$\mh$,
\begin{itemize}
\item[(a)] if~$\mh$ is FT-tight then all FLP-stable models of~$\mh$ are
FT-stable;
\item[(b)] if~$\mh$ is FLP-tight then all FT-stable models of~$\mh$ are
FLP-stable.
\end{itemize}
\medskip

The proof of Main Lemma is given in the appendix.  Some parts of the
proof are inspired by results from~\cite{fer06}.

\section{Proof of Main Theorem}

Consider a program~$\Pi$ in the programming language described at the beginning
of this paper.
Every formula in~$\tau\Pi$
corresponds to one of the rules~(\ref{eq:rule}) of~$\Pi$ and has the form
\beq
\tau((B_1,\dots,B_n)^{\bf x}_{\bf r}) \rar
\tau{(H_1)}^{\bf x}_{\bf r}\,\lor\,\cdots\,\lor\,
\tau{(H_k)}^{\bf x}_{\bf r}
\eeq{trr}
where $\bf x$ is the list of all global variables of the rule, and $\bf r$
is a tuple of precomputed terms such that all literals
${(H_i)}^{\bf x}_{\bf r}$, ${(B_j)}^{\bf x}_{\bf r}$ are well-formed.
The consequent of~(\ref{trr}) is a disjunction of atoms over the
signature~$\sigma_0$---the set of
atoms of the form $p({\bf t})$, where $p$ is a symbolic constant and ${\bf t}$
is a tuple of precomputed terms.  The antecedent of~(\ref{trr}) is a
conjunction of formulas of three types:
\begin{itemize}
\item[(i)] literals over~$\sigma_0$---each of them is
$\tau(L^{\bf x}_{\bf r})$ for some symbolic literal~$L$ from the body of
the rule;
\item[(ii)] implications of form~(\ref{eq:agdef})---each of them is
$\tau(E^{\bf x}_{\bf r})$ for some aggregate atom~$E$ from the body of
the rule;
\item[(iii)] negations of such implications---each of them is
$\neg\tau(E^{\bf x}_{\bf r})$ for some aggregate literal~$\no\ E$
from the body of the rule.
\end{itemize}
Each of the formulas $\tau({\bf L}^{{\bf x}}_{\bf r})$ in~(\ref{eq:agdef})
is a conjunction of literals over~$\sigma_0$.  It follows
that~(\ref{eq:agdef}) can be represented in the form
\beq
(\ma_1)^\land \land \bigwedge_{p\in\ma_2} \neg p \,\rar\, \mc^\lor,
\eeq{agsr}
where $\ma_1$ and and $\ma_2$ are sets of atoms from~$\sigma_0$, and $\mc$ is
a set of conjunctions of literals over~$\sigma_0$.

Consider the simple program~$\mh$ obtained from $\tau\Pi$ by transforming
the conjunctive terms of the antecedents of its formulas~(\ref{trr}) as
follows:
\begin{itemize}
\item Every literal~$L$ is replaced by the simple implication
\beq
\top\rar L.
\eeq{lit}
\item Every implication (\ref{agsr}) is replaced by the simple
formula
\beq\bigwedge_\phi\left(
(\ma_1)^\land \,\rar\, \bigvee_{p\in\ma_2} \neg\neg p
\,\lor\, \bigvee_{C\in \mc,\ C\ \hbox{\scriptsize is non-empty}}\phi(C)\right),
\eeq{agsrexp}
where the big conjunction extends over all functions~$\phi$ that map every
non-empty conjunction from $\mc$ to one of its conjunctive terms.
\item Every negated implication (\ref{agsr}) is replaced by the simple
formula
\beq
\bigwedge_{p\in \ma_1}(\top\rar p) \;\land\;
\bigwedge_{p\in \ma_2}(\top\rar \neg p) \;\land\;
\bigwedge_{C\in\mc}\,\,\,\bigvee_{L\ \hbox{\scriptsize is a conjunctive term of}\ C}
(\top\rar \neg L).
\eeq{negag}
\end{itemize}
Each conjunctive term of the antecedent of~(\ref{trr}) is equivalent to the
simple formula that replaces it in~$\mh$.  It follows that $\tau\Pi$
and~$\mh$ have the same FLP-stable models.  On the other hand, $\tau\Pi$
and~$\mh$ have the same models in the infinitary logic of here-and-there,
and consequently the same FT-stable models.  Consequently, the FLP-stable
models of~$\Pi$ can be characterized as the FLP-stable models of~$\mh$,
and the FT-stable models of~$\Pi$ can be characterized as the FT-stable
models of~$\mh$.

To derive the main theorem from the main lemma, we will establish two claims that
relate the predicate dependency graph of~$\Pi$ to the dependency graph
of~$\mh$:

\medskip\noindent{\it Claim 1.} If there is an edge from an atom
$p(t_1,\dots,t_k)$ to an atom $q(s_1,\dots,s_l)$ in the dependency graph
of~$\mh$ then there is an edge from~$p/k$ to~$q/l$ in the predicate
dependency graph of~$\Pi$.

\medskip\noindent{\it Claim 2.} If the edge from $p(t_1,\dots,t_k)$ to
$q(s_1,\dots,s_l)$ in the dependency graph~$\mh$ is FT-critical or
FLP-critical then~$\Pi$ contains a rule~(\ref{eq:rule}) such that
\begin{itemize}
\item $p/k$ is the predicate symbol of one of the atoms~$H_i$, and
\item $q/l$ is the predicate symbol of an atom occurring in one of the
       non-positive aggregate literals~$B_j$.
\end{itemize}

\medskip
Using these claims, we will show
that if the dependency graph of~$\mh$ has a path with
infinitely many FT-critical edges or infinitely many FLP-critical edges
then we can find a non-positive aggregate literal recursive with
respect to~$\Pi$.  The assertion of the theorem will immediately follow
then by Main Lemma.

Assume that
$p_1({\bf t}^1),p_2({\bf t}^2),\dots$ is a path in the dependency graph
of~$\mh$ that contains infinitely many FT-critical edges (for FLP-critical
edges, the reasoning is the same).  By Claim~1, the sequence
$p_1/k_1,\,p_2/k_2,\,\dots$, where
$k_i$ is the length of ${\bf t}^i$, is a path in the predicate dependency
graph of~$\Pi$.  Since that graph is finite, there exists a positive
integer~$a$ such that all vertices $p_a/k_a,\,p_{a+1}/k_{a+1},\,\dots$ belong
to the same strongly connected component.
Since the path $p_1({\bf t}^1),p_2({\bf t}^2),\dots$
contains infinitely many FT-critical edges, there exists a~$b\geq a$ such
that the edge from $p_b({\bf t}^b)$ to $p_{b+1}({\bf t}^{b+1})$ is
FT-critical.  By Claim~2, it follows that $\Pi$ contains a
rule~(\ref{eq:rule}) such that $p_b/k_b$ is the predicate symbol of one
of the atoms~$H_i$, and $p_{b+1}/k_{b+1}$ is the predicate symbol of an
atom occurring in one of the non-positive aggregate literals~$B_j$.
Since $p_b/k_b$ and $p_{b+1}/k_{b+1}$ belong to the same strongly
connected component, there exists a path from $p_{b+1}/k_{b+1}$ to
$p_b/k_b$.  It follows that~$B_i$ is recursive with respect to~$\Pi$.

\medskip\noindent{\it Proof of Claim~1.}
If there is an edge from $p(t_1,\dots,t_k)$ to $q(s_1,\dots,s_l)$ in the
dependency graph of~$\mh$ then~$\Pi$ contains a rule~(\ref{eq:rule}) such
that $p(t_1,\dots,t_k)$ occurs in the consequent of one
of the implications~(\ref{trr}) corresponding to this rule, and
$q(s_1,\dots,s_l)$ occurs in one of the formulas
\hbox{(\ref{lit})--(\ref{negag})}.
Then $q(s_1,\dots,s_l)$ occurs also in the antecedent of~(\ref{trr}).  It
follows
that $p/k$ is the predicate symbol of one of the atoms occuring in the
head of the rule, and $q/l$ is the predicate symbol of one of the atoms
occurring in its body.

\medskip\noindent{\it Proof of Claim~2.}
If the edge from $p(t_1,\dots,t_k)$ to $q(s_1,\dots,s_l)$ in the dependency
graph of~$\mh$ is FT-critical then~$\Pi$ contains a rule~(\ref{eq:rule}) such
that $p(t_1,\dots,t_k)$ occurs in the consequent of one
of the implications~(\ref{trr}) corresponding to this rule, and
$q(s_1,\dots,s_l)$ occurs strictly positively in one of the
non-positive conjunctive terms $\mp^\land \rar \mq^\lor$ of one of the
simple conjunctions \hbox{(\ref{lit})--(\ref{negag})}.  If a formula of
form~(\ref{lit}) is non-positive then no atoms occur in it
strictly positively.  Consequently $\mp^\land \rar \mq^\lor$ is a
conjunctive term of one of the formulas~(\ref{agsrexp}) or~(\ref{negag}),
and it corresponds to an aggregate literal from the body of the rule.
That aggregate literal is not positive, because for any positive literal~$E$
no conjunctive term of the corresponding simple conjunction~(\ref{agsrexp})
is non-positive.  It follows that $p/k$ is the predicate symbol of one of
the atoms in the head of the rule, and $q/l$ is the predicate symbol of
an atom from a non-positive aggregate literal in the body.

For FLP-critical edges the reasoning is similar, using the fact that
formulas of form~(\ref{lit}) do not contain double negations, and
neither do formulas of form~(\ref{agsrexp}) corresponding to positive
aggregate literals.

\section{Related Work}

The equivalence between the FLP and FT approaches to defining stable models
for programs without aggregates was established by \cite{fab04},
Theorem~3.  The fact
that this equivalence is not destroyed by the use of positive aggregates was
proved by \cite{fer05}, Theorem~3.  That result is further generalized by
\cite{bar11a}, Theorem 7.

The program
$$\ba l
q(\oo 1),\\
r \ar \ii{count} \{X : \no\ p(X),\ q(X)\} = \oo 1
\ea$$
has no recursive aggregates but is not covered by any of the results
quoted above because it contains a negative literal in the conditions of an
aggregate atom.

\section{Conclusion}

An oversight in the semantics proposed in the ASP-Core document can be
corrected using a translation into the language of infinitary
propositional formulas.  The main theorem of this paper describes
conditions when stable
models in the sense of the (corrected) ASP-Core definition are identical
to stable models in the sense of the input language of~{\sc clingo}.

The main lemma asserts that if a set of infinitary propositional formulas
is FT-tight then its FLP-stable models are FT-stable, and if it is FLP-tight
then its FT-stable models are FLP-stable.

\section*{Acknowledgements}
Martin Gebser made a valuable contribution to our work by pointing out an
oversight in an earlier version of the proof and suggesting a way to correct
it.  We are grateful to Wolfgang Faber, Jorge Fandi\~no,
Michael Gelfond, and Yuanlin Zhang
for useful discussions related to the topic of this paper, and to the anonymous
referees for their comments.  This research was partially supported
by the National Science Foundation under Grant IIS-1422455.

\bibliography{bib}

\newpage
\appendix

\section{Proof of Main Lemma}

If $F$ is a simple disjunction
and $X$ is a set of atoms, by $F^X_\bot$
we denote the simple disjunction obtained from $F$ by removing all disjunctive
terms that belong to~$X$.\footnote{This notation is motivated by the
fact that $F^X_\bot$ is the result of substituting~$\bot$ for
the disjunctive members of~$F$ that belong to~$X$, rewritten as a
simple disjunction.}
If $F$ is a simple implication $\mp^\land \rar \mq^\lor$ then by $F^X_\bot$ we
denote $F$ itself if $\mp \cap X$ is non-empty, and $\mp^\land \rar
(\mq^\lor)^X_\bot$ otherwise.\footnote{This operation is a special case of the NES operation defined by
\cite{fer06}. Distinguishing between the two cases in the definition is crucial
for Lemmas~\ref{lem:bottoredbot} and \ref{lem:redtoredbot}.}
If~$F$ is a simple formula then
$F^X_\bot$ stands for the simple formula obtained by applying this
transformation to all conjunctive terms of~$F$.
It is clear that $F^X_\bot$ entails~$F$.

For any simple program~$\mh$,
by $\mh^X_\bot$ we denote the simple program obtained from $\mh$ by
applying this transformation to~$G$ and $H$ for each
simple rule $G\rar H$ in $\mh$.

\begin{lemma}\label{lem:subbot}
 Let $I$ be a model of a simple program $\mh$, $X$ be
a set of atoms, and $K$ be a subset of $X$ such that the dependency graph
of $\mh$ has no edges from atoms in $K$ to atoms in $X \setminus K$.
If $I$ satisfies $\mh^X_\bot$, then $I$ satisfies~$\mh^K_\bot$.
\end{lemma}

\begin{proof}
Assume on the contrary that $I$ does not satisfy $\mh^K_\bot$. Then there
is a simple rule $G \rar H$ in $\mh$ such that $I$ satisfies $G^K_\bot$ but
does not satisfy $H^K_\bot$. Further, since $I$ satisfies
$G^K_\bot$ and $G^K_\bot$ entails $G$,~$I$ satisfies
$G$ as well. Then since~$I$ is a model of $\mh$,~$I$ satisfies $H$. Since~$I$ satisfies $H$ but does
not satisfy $H^K_\bot$, there is some atom $p$ in~$H$ that is also in $K$.
Now, since~$I$ satisfies $G^K_\bot$ it must also satisfy $G^X_\bot$.
Indeed, if this were not the case, there would be some atom~$q$ occurring
positively in $G$ and also occurring in $X \setminus K$. Then there would be an edge from $p \in K$
to $q \in X \setminus K$, contradicting the assumption of the lemma.
On the other hand, $I$ does not satisfy
$H^X_\bot$, since $I$ does not satisfy $H^K_\bot$ and $K$ is a subset of~$X$.
We may conclude that $I$ does not satisfy $G^X_\bot \rar H^X_\bot$
and therefore does not satisfy $\mh^X_\bot$.
\end{proof}

\begin{lemma}\label{lem:lnsft}
Let $I$ be a model of a simple program $\mh$ and let $K$ be a subset of $I$ such
that there are no FT-critical edges in the subgraph of the dependency graph of $\mh$ induced by $K$.
If $I \models \mh^K_\bot$ then
$I \setminus K$ satisfies the FLP-reduct of $\mh$ with respect to $I$.

\end{lemma}
\begin{proof}
Consider a simple rule $G \rar H$ in $\mh$ such that $I \models G$, so that
$G \rar H$
is in the FLP-reduct of $\mh$. We will show that $I\setminus K$ satisfies
$G \rar H$. Since $I
\models \mh^K_\bot$, either $I \not \models G^K_\bot$ or $I \models H^K_\bot$.

\medskip

\noindent {\it Case 1: $I \models H^K_\bot$.}  Then~$H$ has a disjunctive
term that belongs to~$I$ but not to~$K$, so that \hbox{$I \setminus K
\models H$}. We conclude that $I \setminus K \models G \rar H$.

\medskip

\noindent {\it Case 2: $I \not \models G^K_\bot$.} Consider a
conjunctive term~$\mp^\land \rar \mq^\lor$ in $G$ such that $I \not \models
(\mp^\land \rar \mq^\lor)^K_\bot$.
Since $I \models G$, $I \models \mp^\land \rar \mq^\lor$. It follows that $\mp
\cap K$ is empty and that $I$ satisfies both $\mp^\land$ and~$\mq^\lor$ but
does not satisfy $(\mq^\lor)^K_\bot$.

\medskip

\noindent {\it Case 2.1: $\ma^\land \rar \mq^\lor$ is positive.}
Then $\mq$ is a set of atoms.  Since $I \not \models (\mq^\lor)^K_\bot$,
all atoms from $I$ that are in $\mq$ are also in $K$. So $I \setminus
K \not \models \mq^\lor$.
Since $\mp \cap K$ is empty
and $I$ satisfies $\mp^\land$, $I \setminus K$ also satisfies $\mp^\land$. We may
conclude that $I \setminus K \not \models G$ so that $I \setminus K \models G \rar H$.

\medskip

\noindent {\it Case 2.2: $\ma^\land \rar \mq^\lor$ is non-positive.}
Since $I$ satisfies $\mq^\lor$ but not $(\mq^\lor)^K_\bot$, there is an
atomic disjunctive term $p$ in~$\mq$ that belongs to $I \cap K$.
Then $p$ occurs positively in $G$.
It follows that no disjunctive term in $H$ occurs in $K$. (If there were such a
disjunctive term $q$ in $H$ then, since $\ma^\land \rar \mq^\lor$ is non-positive, there
would be an FT-critical edge from $q$ to $p$ in the subgraph of the dependency
graph of $\mh$ induced by $K$.
But the condition
of the lemma stipulates that there are no FT-critical edges in
that graph.)
Since $I$ satisfies $G$ and is a model of the program, $I$ satisfies $H$
as well.  Since no atoms from $K$ occur in $H$, it follows that
$I \setminus K$ satisfies $H$, so that $I \setminus K$
satisfies $G \rar H$.
\end{proof}

\begin{lemma}\label{lem:findK}
If $\mh$ is an FT-tight simple program and $X$ is a non-empty set of atoms, then there
exists a non-empty subset $K$ of $X$ such that in the subgraph of the dependency
graph of $\mh$ induced by~$X$
\begin{enumerate}
\item[(i)] there are no edges from $K$ to atoms in $X \setminus K$, and
\item[(ii)] no atom in $K$ has outgoing FT-critical edges.
\end{enumerate}
\end{lemma}

\begin{proof}
Consider the subgraph of the dependency graph of $\mh$ induced by $X$.
It contains some vertex $b$ such that there is no path starting at
$b$ that contains an FT-critical edge.
(If there were no such vertex $b$, then there would be a path
containing infinitely many FT-critical edges and $\mh$ would not be FT-tight.)
Take $K$ to be the set of all vertices
reachable from $b$. It is clear that condition (i) is satisfied.
Furthermore, since all atoms in $K$ are reachable from $b$, and no path starting
at~$b$ contains an FT-critical edge, none of the atoms in $K$ have outgoing
FT-critical edges in the subgraph of the dependency graph of $\mh$ induced by
$X$. So condition (ii) is satisfied as~well.
\end{proof}

\begin{lemma}\label{lem:dtosft}
If $\mh$ is an FT-tight simple program and $I$ is an FLP-stable model of $\mh$, then
for every non-empty subset $X$ of $I$, $I \not
\models \mh^X_\bot$.
\end{lemma}

\begin{proof}
Assume on the contrary that there is some non-empty subset $X$ of $I$ such that $I
\models \mh^X_\bot$.
By Lemma~\ref{lem:findK}, there is a non-empty subset~$K$ of~$X$ meeting
conditions (i) and (ii).
Since $I \models \mh^X_\bot$ and $K$ satisfies (i), by
Lemma~\ref{lem:subbot} we may conclude that $I \models \mh^K_\bot$.
Since $K$ satisfies condition~(ii) and is a subset of $X$, it is clear that there are no FT-critical
edges in the subgraph of the dependency graph of $\mh$ induced by $K$. So by Lemma~\ref{lem:lnsft}, $I \setminus K$
satisfies the FLP-reduct of $\mh$,
contradicting the assumption that $I$ is FLP-stable.
\end{proof}

\begin{lemma}\label{lem:bottoredbot}
Let $G$ be a simple disjunction or a simple formula, and let~$X$
be a set of atoms.  An interpretation~$I$ satisfies
$G^X_\bot$ iff it satisfies $FT(G, I)^X_\bot$.
\end{lemma}

\begin{proof}
To prove the assertion for a simple disjunction, it is sufficient to consider
the case when $G$ is a single extended literal.
If $G$ is an atom~$p$,
$$I\models p^X_\bot
\quad\hbox{iff}\quad
p \in I \hbox{ and }p \not \in X
\quad\hbox{iff}\quad
I\models FT(p,I)^X_\bot.$$
If $G$ is either $\neg p$ or $\neg \neg p$, then $G^X_\bot$ is $G$ and
$FT(G, I)^X_\bot$ is $FT(G, I)$. It is clear that $I$ satisfies
$G$ iff it satisfies $FT(G, I)$.

To prove the assertion of the lemma for simple formulas, it is sufficient
to consider the case when~$G$ is a single simple implication $\mp^\land
\rar \mq^\lor$.
If $I$ does not satisfy $G$ then it does not satisfy $G^X_\bot$ either;
on the other hand, in this case
$FT(G, I)$ is~$\bot$, and so is $FT(G, I)^X_\bot$.
Otherwise, $FT(G, I)^X_\bot$ is
\beq
\left ( FT(\mp^\land, I) \rar FT(\mq^\lor, I) \right )^X_\bot.
\eeq{eq:redxbot}
We consider two cases corresponding to whether or not $\mp \cap X
\cap I$ is
empty. If $\mp \cap X \cap I$ is empty, $I$ does not satisfy \eqref{eq:redxbot}
iff $$\mp \subseteq I \quad \hbox{ and } \quad I \not \models FT(\mq^\lor,
I)^X_\bot,$$ or equivalently,
$$I \models \mp^\land \quad \hbox{ and } \quad I \not \models (\mq^\lor)^X_\bot.$$
If on the other hand, $\mp \cap X \cap I$ is non-empty, then \eqref{eq:redxbot} is
$FT(G, I)$ and $G^X_\bot$ is $G$.
\end{proof}

\medskip

\begin{lemma}\label{lem:redtoredbot}
For any simple disjunction  $G$ and any interpretations $I$ and $J$,
\begin{center}
$J\models FT(G, I)$ iff $I\models FT(G, I)^{I \setminus J}_\bot$.
\end{center}
\end{lemma}

\begin{proof}
It is sufficient to prove the lemma for the case when $G$ is a single extended literal.
If $G$ is an atom $p$ then
$$J\models FT(p, I)
\quad\hbox{iff}\quad
p \in I\hbox{ and }p \in J
\quad\hbox{iff}\quad
I\models FT(p,I)^{I\setminus J}_\bot.$$
If $G$ is $\neg p$ then both
$$J\models FT(G, I)$$ and
$$I\models FT(G, I)^{I \setminus J}_\bot$$
are equivalent to $p \not\in I$.
If $G$ is $\neg \neg p$ then both $$J \models FT(G, I)$$
and
$$I\models FT(G, I)^{I \setminus J}_\bot$$ are equivalent to
$p \in I$.
\end{proof}

\begin{lemma}\label{lem:redtoredbotimp}
For any simple formula $G$ and any interpretations $I$ and $J$,
if $I\models FT(G, I)^{I \setminus J}_\bot$ then \hbox{$J\models FT(G, I)$}.
\end{lemma}

\begin{proof}
It is sufficient
to consider the case when~$G$ is a single simple implication $\mp^\land
\rar \mq^\lor$.
If $I$ does not satisfy~$G$ then both $FT(G, I)$ and $FT(G, I)^{I \setminus
J}_\bot$ are $\bot$.
Assume $I$ satisfies $G$. Then $FT(G, I)$ is
$$
FT(\mp^\land, I) \rar FT(\mq^\lor, I).
$$
We consider two cases corresponding to whether or not
$\mp \cap I\setminus J$ is
empty. If $\mp \cap I \setminus J$ is non-empty,
then $FT(G, I)^{I \setminus J}_\bot$ is $FT(G, I)$. Furthermore, $J$ does not
satisfy $FT(\mp^\land, I)$. Indeed, if it did, $\mp$ would be a subset of both
$I$ and $J$, contradicting the assumption that $\mp \cap I \setminus J$ is
non-empty. It follows that $J \models FT(G, I)$.
If, on the other hand, $\mp \cap I \setminus J$ is empty,
then $FT(G, I)^{I \setminus J}_\bot$ is
$$
FT(\mp^\land, I) \rar \left (FT(\mq^\lor, I) \right )^{I \setminus J}_\bot.
$$
Assume that $J$ does not satisfy $FT(G, I)$. Then
$$J \models FT(\mp^\land, I) \quad \hbox{ and } \quad J \not \models FT(\mq^\lor, I).
$$
From the first condition we may conclude that $I \models FT(\mp^\land, I)$.
(Indeed, if $J \models FT(\mp^\land, I)$ then~$\mp$ must be a subset both of $I$
and of $J$.)
From the last condition using Lemma~\ref{lem:redtoredbot} it follows that
$I \not \models FT(\mq^\lor, I)^{I \setminus J}_\bot.$ We may conclude that $I
\not \models FT(G, I)^{I\setminus J}_\bot$.
\end{proof}

\medskip\noindent{\em Proof of Part (a) of Main Lemma}

\noindent
Let $I$ be an FLP-stable model of an FT-tight simple program~$\mh$.
Then $I \models \mh$, so that $I \models FT(\mh, I)$. We need to
show that no proper subset~$J$ of~$I$ satisfies $FT(\mh, I)$. Take a
proper subset $J$ of $I$, and let $X$ be
$I\setminus J$. By Lemma~\ref{lem:dtosft}, $I$ does not satisfy $\mh^X_\bot$. Then
there is a simple rule $G \rar H$ in $\mh$ such that $I$ satisfies
$G^X_\bot$ and does not satisfy $H^X_\bot$.
By Lemma~\ref{lem:bottoredbot}, it follows that~$I$ satisfies
$FT(G, I)^X_\bot$ and does not satisfy $FT(H, I)^X_\bot$.
Since $X = I \setminus J$, it follows that $J$ satisfies $FT(G, I)$
(Lemma~\ref{lem:redtoredbotimp}) but does not satisfy $FT(H, I)$
(Lemma~\ref{lem:redtoredbot}).  So $J$ does not satisfy
$FT(\mh, I)$. It follows that $I$ is FT-stable.
\medskip

We turn now to the proof of part~(b) of Main Lemma.
If $F$ is a simple disjunction then by $F^+$ we denote
the result of replacing each extended literal $\neg \neg p$ in $F$ by $p$,
and similarly for simple implications, formulas, rules, and programs.

\begin{lemma}\label{lem:lnsflp}
Let $I$ be a model of a simple program $\mh$, and let $K$ be a set of atoms
such that
there are no FLP-critical edges in the subgraph of the dependency graph of $\mh$ induced by $K$.
If $I \models (\mh^+)^K_\bot$ then
$I \setminus K$ satisfies the FT-reduct of $\mh$ with respect to $I$.
\end{lemma}
\begin{proof}
We need to show that $I \setminus K$ satisfies the FT-reduct of every
simple rule $G \rar H$ in~$\mh$.  Since $I$ is a model
of~$\mh$, that reduct is $FT(G, I) \rar FT(H, I)$.  If $I \not \models G$ then
the antecedent of this implication is equivalent to $\bot$, and the
assertion that the implication is satisified by $I \setminus K$ is trivial.

Assume then that  $I \models G$.  Since $I$ is a model of $\mh$, it
follows that $I\models H$. Since $I \models (\mh^+)^K_\bot$, either
$I \not \models (G^+)^K_\bot$ or $I \models H^K_\bot$.
\medskip

\noindent {\it Case 1: $I \models H^K_\bot$.}  Then~$H$ has a disjunctive
term~$p$ that belongs to~$I$ but not to~$K$. Then $p$
is also a disjunctive term in $FT(H, I)$, so that $I \setminus K
\models FT(H, I)$. We conclude that $I \setminus K$ satisfies
$FT(G \rar H, I)$.

\medskip
\noindent {\it Case 2: $I \not \models (G^+)^K_\bot$.} Consider a
conjunctive term~$$\mp^\land \rar \mq^\lor$$ in $G$ such that $I$ does not satisfy
$(\mp^\land \rar (\mq^\lor)^+)^K_\bot$.
Since $I \models G$, $I \models \mp^\land \rar \mq^\lor$. It follows that $\mp
\cap K$ is empty and that $I$ satisfies $\mp^\land$, $\mq^\lor$
and $(\mq^\lor)^+$ but does not satisfy $((\mq^\lor)^+)^K_\bot$.

\medskip

\noindent {\it Case 2.1: $\mq^\lor$  does not contain any extended
literal $\neg \neg p$ such that $p \in K$.}
Since $I$ satisfies $(\mq^\lor)^+$ but not $((\mq^\lor)^+)^K_\bot$, each atomic
disjunctive term
$p$ in~$(\mq^\lor)^+$ that is in $I$ must also be in $K$.
Furthermore, $I$ cannot satisfy any literal $\neg p$ in $\mq$. (If it did, then
that literal would also be in $((\mq^\lor)^+)^K_\bot$, and this
disjunction would be satisfied by~$I$.)
 Since $\mq$ does not contain any extended literal $\neg\neg p$
 such that $p$ is in $K$, $I$ does not satisfy any extended literal
$\neg \neg p$
 in $\mq$. (For each extended literal $\neg \neg p$ in $\mq$, $p$ is a
 disjunctive term in $((\mq^\lor)^+)^K_\bot$. If~$I$ satisfied some
extended
 literal $\neg \neg p \in \mq$, then $I$ would satisfy $p$ and therefore
also
 satisfy $((\mq^\lor)^+)^K_\bot$.)  We conclude that every extended literal
in~$\mq$ that is satisfied by~$I$ is an atom from~$K$.  It follows that
$FT(\mq^\lor, I)$ is  equivalent to a disjunction of atoms
from~$K$.
So $I \setminus
K \not \models FT(\mq^\lor, I)$. Since $I \models \mp^\land$, $I \models
FT(\mp^\land, I)$.  Since $\mp \cap K$ is empty,
$I \setminus K$ also satisfies $FT(\mp^\land)$. We may
conclude that $I \setminus K \not \models FT(G, I)$ so that $I \setminus K
\models FT(G \rar H, I)$.
\medskip

\noindent {\it Case 2.2: $\mq^\lor$ contains an extended
literal $\neg \neg p$ such that $p \in K$.}
Then no disjunctive term in $H$ occurs in $K$. (If there were such a
disjunctive term $q$ in $H$ then there
would be an FLP-critical edge from $q$ to $p$ in the subgraph of the dependency
graph of $\mh$ induced by $K$.
But the condition
of the lemma stipulates that there are no FLP-critical edges in
that graph.)
Since $I$ satisfies $H$, $I$ satisfies $FT(H, I)$ as well. Since no atoms
from $K$ occur in $H$, it
follows that $I \setminus K$ satisfies $FT(H, I)$, so that $I \setminus K$
satisfies $FT(G \rar H, I)$.
\end{proof}

\begin{lemma}\label{lem:findK1}
If $\mh$ is an FLP-tight simple program and $X$ is a non-empty set of atoms, then there
exists a non-empty subset $K$ of $X$ such that
in the
subgraph of the dependency graph of $\mh$ induced by~$X$
\begin{enumerate}
\item[(i)] there are no edges from $K$ to atoms in $X \setminus K$, and
\item[(ii)]  no atom in $K$ has outgoing FLP-critical edges.
\end{enumerate}
\end{lemma}

\noindent
The proof is similar to the proof of Lemma~\ref{lem:findK}.

\begin{lemma}\label{lem:dtosflp}
If $\mh$ is an FLP-tight simple program and $I$ is an FT-stable model of $\mh$, then
for every non-empty subset $X$ of~$I$, $I \not \models (\mh^+)^X_\bot$.
\end{lemma}

\begin{proof}
Assume on the contrary that $I \models (\mh^+)^X_\bot$
for some non-empty subset $X$ of $I$.  Consider a non-empty subset~$K$
of~$X$ meeting conditions (i) and (ii) from Lemma~\ref{lem:findK1}.
Since $I \models (\mh^+)^X_\bot$ and $K$ satisfies (i), by
Lemma~\ref{lem:subbot} we may conclude that $I \models (\mh^+)^K_\bot$.
Since $K$ satisfies~(ii) and is a subset of $X$, there are no FLP-critical
edges in the subgraph of the dependency graph of $\mh$ induced by $K$. So by Lemma~\ref{lem:lnsflp}, $I \setminus K$
satisfies the FT-reduct of $\mh$,
contradicting the assumption that $I$ is FT-stable.
\end{proof}

\begin{lemma}\label{lem:redtoredbotflp}
Let $G$ be a simple disjunction or a simple formula.  For any interpretations
$I$ and $J$ such that $J \subseteq I$, if $I\models
(G^+)^{I \setminus J}_\bot$ then $J\models G$.
\end{lemma}

\begin{proof}
To prove the assertion of the lemma for simple disjunctions, it is
sufficient to consider the case when $G$ is a single
extended literal. If $G$ is $p$ or $\neg\neg p$ then
$(G^+)^{I \setminus J}_\bot$ is $p^{I \setminus J}_\bot$.  Since~$I$
satisfies this formula, $p \in J$, so that $J \models G$.
If~$G$ is~$\neg p$ then $(G^+)^{I \setminus J}_\bot$ is~$\neg p$.
Since $I\models \neg p$ and $J \subseteq I$, $J \models \neg p$.

To prove the assertion of the lemma for simple formulas, it is sufficient
to consider the case when~$G$ is a single simple implication $\mp^\land
\rar \mq^\lor$.  If $\mp \cap I \setminus J$ is non-empty
then $$J \not \models \mp^\land,$$ so that $J \models G$.
If, on the other hand, $\mp \cap I \setminus J$ is empty
then $(G^+)^{I \setminus J}_\bot$ is
$\mp^\land \rar ((\mq^\lor)^+)^{I \setminus J}_\bot$.
Assume that $J$ does not satisfy $G$. Then
$$J \models \mp^\land \quad \hbox{ and } \quad J \not \models \mq^\lor.$$
From the first condition and the fact that $J \subseteq I$ we may conclude
that $I \models \mp^\land$.  From the second condition it follows, by
the part of the lemma proved above, that
$I \not \models ((\mq^\lor)^+)^{I \setminus J}_\bot.$  Consequently
$I \not \models (G^+)^{I\setminus J}_\bot$.
\end{proof}

\medskip\noindent{\em Proof of Part (b) of Main Lemma}

\noindent
Let $I$ be an FT-stable model of an FLP-tight simple program~$\mh$. Then
$I$ is a model of $\mh$, and consequently a model of the reduct
$\FLP(\mh,I)$.
We need to show that no proper subset $J$ of $I$ is a model of this reduct.
Consider a proper subset~$J$ of~$I$, and let $X$ be
$I\setminus J$. By
Lemma~\ref{lem:dtosflp}, $I$ does not satisfy $(\mh^+)^X_\bot$. Then
there is a simple rule $G \rar H$ in $\mh$ such that $I$ satisfies
$(G^+)^X_\bot$ and does not satisfy $H^X_\bot$.
Since $(G^+)^X_\bot$ entails $G^+$, and $G^+$ is  equivalent
to~$G$, we can conclude that~$I$ satisfies~$G$, so that $G \rar H$ belongs to
the reduct $\FLP(\mh, I)$.  On the other hand,
by Lemma~\ref{lem:redtoredbotflp},~$J$ satisfies $G$.
Since $I$ does not satisfy $H^{X}_\bot$ and $H$ is a disjunction of atoms, $J$ does not satisfy $H$.
So $J$ does not satisfy $G \rar H$, and consequently is not a model
$\FLP(\mh, I)$.

\bibliographystyle{acmtrans}
\end{document}